\documentclass[conference,letter]{IEEEtran}

\usepackage[utf8]{inputenc} 
\usepackage[T1]{fontenc}
\usepackage{url}
\usepackage{ifthen}
\usepackage{cite}
\usepackage{amsmath}%
\usepackage{wasysym}%
\usepackage{amssymb}

\usepackage{mdwmath}
\usepackage{blindtext}
\usepackage{eqparbox}
\DeclareMathOperator{\E}{\mathbb{E}}

\usepackage[english]{babel}
\usepackage{lipsum}
\usepackage{xcolor}
\usepackage{tikz}
\usepackage{graphicx}
\usepackage{algorithm, algorithmic}
\usepackage{siunitx}

\usepackage{calc}
\newcolumntype{M}[1]{>{\centering\arraybackslash}m{#1}}
\newcolumntype{N}{@{}m{0pt}@{}}

\usepackage{geometry}
\usepackage{amsthm}

\newtheorem{theorem}{{Theorem}}

\newtheorem{proposition}[theorem]{{Proposition}}

\DeclareMathAlphabet{\mathbfsl}{OT1}{ppl}{b}{it} 

\newcommand{\be}[1]{\begin{equation}\label{#1}}
\newcommand{\ee}{\end{equation}}

\renewcommand{\leq}{\leqslant}
 
\renewcommand{\geq}{\geqslant}

\newcommand{\Tref}[1]{Theo\-rem\,\ref{#1}}
\newcommand{\Pref}[1]{Pro\-po\-si\-tion\,\ref{#1}}

\newcommand{\Cref}[1]{Co\-ro\-lla\-ry\,\ref{#1}}

\IEEEoverridecommandlockouts
\begin{document}
\setlength{\abovedisplayskip}{4pt}
	\setlength{\belowdisplayskip}{4pt}
\title{Covert Millimeter-Wave Communication via a Dual-Beam Transmitter\vspace{-0.3cm}} 

 \author{
\IEEEauthorblockN{Mohammad Vahid Jamali and Hessam Mahdavifar}
   \IEEEauthorblockA{Department of Electrical Engineering and Computer Science, University of Michigan, Ann Arbor, MI 48109, USA\\
   	E-mails: mvjamali@umich.edu, hessam@umich.edu
   }
   \thanks{This work was supported by the National Science Foundation
under grant CCF--1763348.}    
}

 \maketitle
 
\begin{abstract}
In this paper, we investigate covert communication over millimeter-wave (mmWave) frequencies. In particular, a dual-beam mmWave transmitter, comprised of two independent antenna arrays, attempts to reliably communicate to a receiver Bob when hiding the existence of transmission from a warden Willie. In this regard, operating over mmWave bands not only increases the covertness thanks to directional beams, but also increases the transmission data rates given much more available bandwidths and enables ultra-low form factor transceivers due to the lower wavelengths used compared to the conventional radio frequency (RF) counterpart. We assume that the transmitter Alice employs one of its antenna arrays to form a directive beam for transmission to Bob. The other antenna array is used by Alice to generate another beam toward Willie as a jamming signal with its transmit power changing independently from a transmission block to another block. We characterize Willie's detection performance with the optimal detector and the closed-form of its expected value from Alice's perspective. We further derive the closed-form expression for the outage probability of the Alice-Bob link, which enables characterizing the optimal covert rate that can be achieved using the proposed setup. Our results demonstrate the superiority of mmWave covert communication, in terms of covertness and rate, compared to the RF counterpart.
\end{abstract}

\section{Introduction}
Rapid growth of wireless networks and emerging variety of applications, including Internet of Things (IoT) and critical controls, necessitate sophisticated solutions to secure the data transmission. Traditionally, the main objective of security schemes, using either cryptographic or information-theoretic approaches, is to secure data in the presence of adversary eavesdroppers. However, a stronger level of security can be obtained in wireless networks if the existence of communication is hidden from the adversaries. To this end, recently, there has been increasing attention to investigate covert communication, also referred to as communication with \textit{low probability of detection} (LPD), in different setups with the goal of hiding the existence of communication \cite{bash2013limits, bash2015hiding, che2013reliable, wang2016fundamental, bloch2016covert,tahmasbi2018first, arumugam2019embedding,arumugam2018covert}.  Generally speaking, covert communication refers to the problem of reliable communication between a transmitter Alice and a receiver Bob while maintaining a low probability of detecting the existence of communication from the perspective of a warden Willie \cite{bash2015hiding}.

The information-theoretic limits on the rate of covert communication have been presented in \cite{bash2013limits} for additive white Gaussian noise (AWGN) channels, where it is proved that $\mathcal{O}(\sqrt{n})$ bits of information can be transmitted to Bob, reliably and covertly, in $n\to\infty$ uses of the channel. The same square root law have been developed for binary symmetric channels (BSCs) \cite{che2013reliable} and  discrete memoryless channels (DMCs) \cite{wang2016fundamental}. Moreover, the principle of channel resolvability has been  used in \cite{bloch2016covert} to develop a coding scheme that can reduce the number of required shared key bits. Also, the first- and second-order asymptotics
of covert communication over binary-input discrete memoryless channels are studied in \cite{tahmasbi2018first}. The covert communications setup is also extended to a broadcast scenario \cite{arumugam2019embedding} and mutiple-access scenarios \cite{arumugam2018covert} from an information-theoretic perspective. 

The achievable covert rate in the aforementioned framework is zero in the limit of large $n$ when noting that $\lim_{n\to\infty}\mathcal{O}(\sqrt{n})/n=0$. However, it is demonstrated that positive covert rates can be achieved by introducing additional uncertainty, from Willie's perspective, into the system. For instance, it is shown in \cite{lee2015achieving,goeckel2016covert} that Willie's uncertainty about his noise power helps in achieving positive 
covert rates. Moreover, by considering slotted AWGN channels, it is proved in \cite{bash2016covert} that positive covert rates are achievable if the warden does not know when the transmission is taking place. The possibility of achieving positive-rate covert communication is further demonstrated for amplify-and-forward (AF) relaying networks with a greedy relay attempting to transmit its own information to the destination on top of forwarding the source's information \cite{hu2018covert}, dual-hop relaying systems with channel uncertainty \cite{wang2019covert}, a downlink scenario under channel uncertainty and with a legitimate user as cover \cite{shahzad2017covert}, and a single-hop setup with a full-duplex receiver acting as a jammer \cite{hu2018covert2}. 

Prior studies on covert communication in wireless setups often consider transmission over conventional radio frequency (RF) wireless links. However, a superior performance, in terms of covertness and achievable rates, can be attained when performing the covert communication over millimeter-wave (mmWave) bands. This makes mmWave bands a suitable option for covert communication to increase the security level of wireless applications involving critical data, motivating the investigation of \textit{mmWave covert communications}. Given the significant differences between the channel models and system architectures of mmWave communication and that of RF communication, the existing results on covert communications can not be immediately extended to covert communication over mmWave band. In particular, communication over mmWave bands can exploit directive beams, thanks to the deployment of massive antenna arrays, to compensate for the path losses over these ranges of frequencies. In this paper, we study covert communication over mmWave channels from a communication theory perspective, i.e., by analyzing the performance of the system in the limit of large block-lengths $n\to\infty$, and show the superiority of mmWave covert communication, in terms of covertness and rate, compared to the RF counterpart.

\section{Channel and System Models}
\subsection{Channel Model}
Recent experimental studies have demonstrated that mmWave links are highly sensitive to blocking effects \cite{rappaport2013millimeter,andrews2017modeling}. In order to model this characteristic, a proper channel model should differentiate between the LOS and non-LOS (NLOS) channel models since the path loss in the NLOS links can be much higher than that of the LOS path due to the weak diffractions in the mmWave band. Therefore, two different sets of parameters are considered for the LOS and NLOS mmWave links, and a deterministic function $P_{\rm LOS}(d_{ij}) \in [0,1]$, that is a non-increasing function of $d_{ij}$, the link distance (in meters) between nodes $i$ and $j$, is the probability that an arbitrary link of length $d_{ij}$ is LOS. We consider a generic $P_{\rm LOS}(d_{ij})$ throughout our analysis and use the model $P_{\rm LOS}(d_{ij})={\rm e}^{-d_{ij}/200}$, suggested in \cite{andrews2017modeling}, for our numerical analysis.

Similar to \cite{andrews2017modeling}, we express the channel coefficient for an arbitrary mmWave link between the transmitter $i$ and receiver $j$ as $h_{ij}=\tilde{h}_{ij}\sqrt{{G}_{ij}L_{ij}}$, where $\tilde{h}_{ij}$, ${G}_{ij}$, and $L_{ij}$ are respectively the channel fading coefficient, the total directivity gain (including both the transmitter and receiver beamforming gains), and the path loss of the $i$-$j$ mmWave link.


To characterize the path loss $L_{ij}$ of the $i-j$ link with the length $d_{ij}$, we consider different path loss exponents ($\alpha_{\rm L},\alpha_{\rm N}$) and intercepts ($C_{\rm L},C_{\rm N}$) for the LOS and NLOS links, respectively. Consequently, denoting the path losses associated with the LOS and NLOS links as $L^{(\rm L)}_{ij}$ and $L^{(\rm N)}_{ij}$, respectively, the path loss $L_{ij}$ is either $L^{(\rm L)}_{ij}=C_{\rm L}d_{ij}^{-\alpha_{\rm L}}$ with probability $P_{\rm LOS}(d_{ij})$ or $L^{(\rm N)}_{ij}=C_{\rm N}d_{ij}^{-\alpha_{\rm N}}$ with probability $1-P_{\rm LOS}(d_{ij})$.

To ascertain the total directivity gain ${G}_{ij}$, we use the common sectored-pattern antenna model adopted in \cite{bai2015coverage,di2015stochastic} which approximates the actual array beam pattern by a step function, i.e., with a constant main lobe $M^{(q)}_{\rm X}$ over the beamwidth $\theta^{(q)}_{\rm X}$ and a constant side lobe $m^{(q)}_{\rm X}$ otherwise, where ${\rm X}\in\{{\rm TX,RX}\}$ and ${q}\in\{{i,j}\}$.
Then, for a given link, if the spatial arrangement of the beams of the transmitter and receiver are known, the total directivity gain can be known from the product of the gains of the transmitter and receiver. In a particular case where the main lobe of a node $q$ (either transmitter or receiver) is pointed to another node, 
we assume some additive beamsteering errors denoted by a symmetric random variable (RV) $\mathcal{E}^{(q)}_{\rm X}$ around the transmitter-receiver direction. Similar to \cite{di2015stochastic}, it is assumed that node $q$ will have a gain of $M^{(q)}_{\rm X}$ if $|\mathcal{E}^{(q)}_{\rm X}|<\theta^{(q)}_{\rm X}/2$, i.e., with probability $F_{|\mathcal{E}^{(q)}_{\rm X}|}(\theta^{(q)}_{\rm X}/2)$ and a gain of $m^{(q)}_{\rm X}$ otherwise, where $F_X(x)$ is the cumulative distribution function (CDF) of the RV $X$. 

Finally, it is common in the literature to model the fading amplitude of mmWave links as independent Nakagami-distributed RVs each with shape parameter $\nu\geq 1/2$ and scale parameter $\Omega=\E[|{\tilde{h}}_{ij}|^2]=1$, and consider different Nakagami parameters for the LOS and NLOS links as $\nu_{\rm L}$ and $\nu_{\rm N}$, respectively \cite{andrews2017modeling,bai2015coverage}. In the case of Nakagami-$m$ fading with parameters $\nu_{\mathcal{B}}$, $\mathcal{B}\in\{{\rm L},{\rm N}\}$, and $\Omega=1$, $|{\tilde{h}}_{ij}|^2$ has a normalized gamma distribution with shape and scale parameters of $\nu_{\mathcal{B}}$ and $1/\nu_{\mathcal{B}}$, respectively, i.e., $|{\tilde{h}}_{ij}|^2\sim{\rm Gamma}(\nu_{\mathcal{B}},1/\nu_{\mathcal{B}})$.

Note that, from an information-theoretic perspective, mmWave communications, and in general wideband communications, can be viewed as low-capacity scenarios \cite{fereydounian2018channel} suggesting a natural framework for mmWave covert communications. 

\subsection{System Model}
We consider the common setup of covert communication comprised of three parties: a transmitter Alice is intending to covertly communicate to a receiver Bob over mmWave bands when a warden Willie is attempting to  realize the existence of this communication. Alice employs a dual-beam mmWave transmitter consisting of two antenna arrays. The first antenna array is used for the transmission to Bob while the second array is exploited as a jammer to enable positive-rate mmWave covert communication. Therefore, when Bob is not in the main lobe of Alice-Willie link, he receives the jamming signal gained with the side lobe of the second array in addition to receiving the desired signal from Alice with the main lobe of the first array. Similarly, when Willie is not in the main lobe of Alice-Bob link, he receives the desired signal gained with the side lobe of the first array in addition to receiving the jamming signal from Alice with the main lobe of the second array. On the other hand, when Bob is in the main lobe of Willie's link, or vice versa, both of their received signals are gained with main lobes. Throughout our analysis we assume that Alice, Bob, and Willie are in some fixed locations (hence having some given directivity gains) and evaluate the impact of changing Willie's location in our numerical results. 


Denoting the channel coefficients between Alice's first and second arrays and the node $j\in\{b,w\}$ (representing Bob and Willie) as $h_{aj,f}$ and $h_{aj,s}$, respectively, one can observe that the path loss gains are the same, i.e., $L_{aj,f}=L_{aj,s}=L_{aj}$, while the fading gains $|\tilde{h}_{aj,f}|^2$ and $|\tilde{h}_{aj,s}|^2$ are independent normalized gamma RVs. We assume quasi-static fading channels meaning that fading coefficients remain constant over a block of $n$ channel uses. We further assume that Alice transmits the desired signal with a publicly-known power $P_a$ while the jamming transmit power $P_J$ of its second array is not known and changes independently from a block to another block. In this paper, we assume that $P_J$ varies according to a uniform distribution in the interval $[0,P_J^{\rm max}]$ while the results can be extended to other distributions using a similar approach. Let $G_{aj,f}$ and $G_{aj,s}$ denote the total directivity gains of the links between Alice's first and second arrays and the node $j\in\{b,w\}$, respectively. Then, when Alice is transmitting to Bob, the received signals by Bob and Willie at each channel use $i=1,2,...,n$ can be expressed as
\begin{align}
\mathbf{y}_{b}(i)=&\sqrt{P_aG_{ab,f}L_{ab}}~\tilde{h}_{ab,f}\mathbf{x}_a(i)\label{yb}\nonumber\\
&+\sqrt{P_JG_{ab,s}L_{ab}}~\tilde{h}_{ab,s}\mathbf{x}_J(i)+\mathbf{n}_b(i),\\
\mathbf{y}_{w}(i)=&\sqrt{P_aG_{aw,f}L_{aw}}~\tilde{h}_{aw,f}\mathbf{x}_a(i)\label{yw}\nonumber\\
&+\sqrt{P_JG_{aw,s}L_{aw}}~\tilde{h}_{aw,s}\mathbf{x}_J(i)+\mathbf{n}_w(i),
\end{align}
respectively, where $\mathbf{x}_a$ and $\mathbf{x}_J$ are Alice's desired and jamming signals, respectively, satisfying $\E[|\mathbf{x}_a(i)|^2]=\E[|\mathbf{x}_J(i)|^2]=1$. Moreover, $\mathbf{n}_b$ and $\mathbf{n}_w$ are zero-mean Gaussian noises at Bob and Willie's receiver with  variances $\sigma^2_b$ and $\sigma^2_w$, respectively. 

\section{Willie's Detection Error Rate}

As discussed earlier, Willie attempts to detect whether Alice is transmitting to Bob or not. We assume that Willie has a perfect knowledge about his channel from Alice and applies binary hypothesis testing when he is unaware of the value of $P_J$. The null hypothesis $\mathbb{H}_0$ states that Alice did not transmit to Bob while the alternative hypothesis $\mathbb{H}_1$ specifies that Alice did transmit to Bob. There are two types of errors that can occur. Willie's decision of hypothesis $\mathbb{H}_1$ when $\mathbb{H}_0$ is true is referred to as a \textit{false alarm} and its probability is denoted by $P_{\rm FA}$. On the other hand, Willie's decision in favor of $\mathbb{H}_0$ when $\mathbb{H}_1$ is true classifies a \textit{missed detection}, whose probability is denoted by $P_{\rm MD}$. Then Willie's overall detection error is denoted by $P_{e,w}\triangleq P_{\rm FA}+P_{\rm MD}$. We say a positive-rate covert communication is possible if for any $\epsilon>0$ there exists a positive-rate communication between Alice and Bob satisfying $P_{e,w}\geq1-\epsilon$ as the number of channel uses $n\to\infty$. 

\begin{figure*}
	\normalsize
	\setcounter{equation}{10}
	\begin{align}\label{Epew}
\!\!\!\!\!\!\!\!	\E[P^*_{e,w}]\!=\!\!\!\!\!\!\!\sum_{\mathcal{B}\in\{{\rm L,N}\}}\!\!\!\!\!\!P_{aw}(\mathcal{B})\!\sum_{k=1}^{2}b_k^{(a,s)}\!\bigg[1\!+\!S(\nu_{\mathcal{B}},g_k^{(a,s)})\bigg]\!\!\times\!\!\bigg[1\!-\!S(\nu_{\mathcal{B}},g_k^{(a,s)})\!+\!\frac{P_am_{a,f}\nu_{\mathcal{B}}^{\nu_{\mathcal{B}}}}{P_J^{\rm max}g_k^{(a,s)}\eta_{\mathcal{B}}\Gamma(\nu_{\mathcal{B}})}\!\sum_{l=1}^{\nu_{\mathcal{B}}}\!\!\binom{\nu_{\mathcal{B}}}{l}\frac{(-1)^{l}}{l}	I(\nu_{\mathcal{B}},l,g_k^{(a,s)})\bigg]\!.\!\!
	\end{align}
	\hrulefill
\vspace{-0.5cm}
\end{figure*}
\subsection{$P_{e,w}$ with the Optimal Detector at Willie}
As it is proved in \cite[Lemma 2]{sobers2017covert} for AWGN channels and also pointed out in \cite[Lemma 1]{shahzad2017covert}, the optimal decision rule that minimizes Willie's detection error
  is given by
\begin{align}\label{radiometer}	\setcounter{equation}{3}
T_w\triangleq\frac{1}{n}\sum_{i=1}^{n}|\mathbf{y}_{w}(i)|^2
\underset{\mathbb{H}_0}{\overset{\mathbb{H}_1}{\gtrless}}\tau,
\end{align}
where $\tau$ is Willie's detection threshold for which we obtain its optimal value/range later in this subsection. Also, using \eqref{yw}, the \textit{strong law of large numbers}, and noting that $n\to\infty$ (see, e.g., the proof of \cite[Theorem 3]{shahzad2017covert}), $T_w$ under hypotheses $\mathbb{H}_0$ and $\mathbb{H}_1$ is given by
\begin{align}
T_w^{\mathbb{H}_0}&=P_JG_{aw,s}L_{aw}|\tilde{h}_{aw,s}|^2+\sigma^2_w,\label{TH0}\\
T_w^{\mathbb{H}_1}&\!=\!P_aG_{aw,f}L_{aw}|\tilde{h}_{aw,f}|^2\!+\!P_JG_{aw,s}L_{aw}|\tilde{h}_{aw,s}|^2\!+\!\sigma^2_w,\label{TH1}
\end{align}
respectively. The optimal threshold of Willie's detector and his corresponding detection error are characterized in the following theorem.
\begin{theorem}\label{Thm1}
The optimal threshold $\tau^*$ for Willie's detector is in the interval
\begin{align}\label{tau*}
\tau^*\in\left\{\begin{matrix}
[\lambda_1,\lambda_2], &\lambda_1<\lambda_2, \\ 
[\lambda_2,\lambda_1], &\lambda_1\geq\lambda_2,
\end{matrix}\right.
\end{align}
and  the corresponding minimum detection error rate is 
\begin{align}\label{pe}
P_{e,w}^*=\left\{\begin{matrix}
\hspace{-2.9cm}0, &\lambda_1<\lambda_2, \\ 
1-\frac{P_aG_{aw,f}|\tilde{h}_{aw,f}|^2}{P_J^{\rm max}G_{aw,s}|\tilde{h}_{aw,s}|^2}, &\lambda_1\geq\lambda_2,
\end{matrix}\right.
\end{align}
where $\lambda_1\triangleq P_J^{\rm max}G_{aw,s}L_{aw}|\tilde{h}_{aw,s}|^2+\sigma^2_w$ and  $\lambda_2\triangleq P_aG_{aw,f}L_{aw}|\tilde{h}_{aw,f}|^2+\sigma^2_w$.
\end{theorem}
\begin{proof}
Using \eqref{TH0} the false alarm probability is given by
\begin{align}\label{pfa}
P_{\rm FA}&=\Pr\left(T_w^{\mathbb{H}_0}>\tau\right)=\Pr\left(P_J>\frac{\tau-\sigma^2_w}{G_{aw,s}L_{aw}|\tilde{h}_{aw,s}|^2}\right)\nonumber\\
&=\left\{\begin{matrix}
\hspace{-3.5cm}1, &\tau<\sigma^2_w, \\ 
1-\frac{\tau-\sigma^2_w}{P_J^{\rm max}G_{aw,s}L_{aw}|\tilde{h}_{aw,s}|^2}, &\sigma^2_w\leq\tau\leq\lambda
_1,\\
\hspace{-3.5cm}0,&\tau\geq\lambda_1.
\end{matrix}\right.
\end{align}
Moreover, using \eqref{TH1} the missed detection probability is given by
\begin{align}\label{pmd}
P_{\rm MD}&=\Pr\left(T_w^{\mathbb{H}_1}<\tau\right)=\Pr\left(P_J<\frac{\tau-\lambda_2}{G_{aw,s}L_{aw}|\tilde{h}_{aw,s}|^2}\right)\nonumber\\
&=\left\{\begin{matrix}
\hspace{-2.8cm}0, &\tau<\lambda_2, \\ 
\frac{\tau-\lambda_2}{P_J^{\rm max}G_{aw,s}L_{aw}|\tilde{h}_{aw,s}|^2}, &\lambda_2\leq\tau\leq\lambda
_3,\\
\hspace{-2.8cm}1,&\tau\geq\lambda_3,
\end{matrix}\right.
\end{align}
where $\lambda_3\triangleq \lambda_2+P_J^{\rm max}G_{aw,s}L_{aw}|\tilde{h}_{aw,s}|^2$. Using \eqref{pfa} and \eqref{pmd}, and by following a similar argument as in the proof of \cite[Theoem 1]{hu2018covert2}, the range of optimal threshold and the corresponding error rate are obtained as \eqref{tau*} and \eqref{pe}, respectively. 
\end{proof}
\noindent\textbf{Remark 1.} Eq. \eqref{pe} shows that for small values of  $P_J^{\rm max}$ such that $P_J^{\rm max}G_{aw,s}|\tilde{h}_{aw,s}|^2\leq P_aG_{aw,f}|\tilde{h}_{aw,f}|^2$ Willie can attain a zero error rate negating the possibility of achieving a positive-rate covert communication in the limit of $n\to\infty$. Although increasing $P_J^{\rm max}$ beyond $P_aG_{aw,f}|\tilde{h}_{aw,f}|^2/(G_{aw,s}|\tilde{h}_{aw,s}|^2)$ can increase $P^*_{e,w}$ and enable a positive-rate covert communication ($P^*_{e,w}\to1$ as $P_J^{\rm max}\to\infty$), it also degrades the performance of the desired Alice-Bob link as we will see in Section IV. The superiority of mmWave covert communication to that of omni-directional  RF communication can be observed by noticing the beneficial impact of beamforming. In fact, in the received signal by Willie, $P_J$ is gained with $G_{aw,s}$ which is much larger than the gain $G_{aw,f}$ of $P_a$; this significantly degrades the performance of Willie's detector. The opposite situation happens for the Alice-Bob link where the desired signal is gained with $G_{ab,f}$ which is much larger than the gain $G_{ab,s}$ of the jamming signal. 

\subsection{$\E[P^*_{e,w}]$ From Alice's Perspective}
Since Alice and Bob are unaware of the instantaneous realization of the channel between Alice and Willie, they should rely on the expected value of $P_{e,w}^*$. Note also that the minimum error rate $P_{e,w}^*$ in \eqref{pe} is independent of the beamforming gain of Willie's receiver as it cancels out in the ratio of $G_{aw,f}/G_{aw,s}$ and also in  the comparison between $\lambda_1$ and $\lambda_2$. Furthermore, Alice perfectly knows the gain $m_{a,f}$ of the side lobe of her first array to Willie. However, she has uncertainty about the gain $g^{(a,s)}$ of the main lobe of the second array toward Willie due to the misalignment error; it is either $g_1^{(a,s)}\triangleq M_{a,s}$ with probability $b_1^{(a,s)}\triangleq F_{|\mathcal{E}_{a,s}|}\left({\theta_{a,s}}/{2}\right)$ or $g_2^{(a,s)}\triangleq m_{a,s}$ with probability $b_2^{(a,s)}\triangleq 1-F_{|\mathcal{E}_{a,s}|}\left({\theta_{a,s}}/{2}\right)$. Moreover, Alice and Bob do not know whether the Alice-Willie link is LOS or NLOS; hence, they should take into account two possibilities given the LOS probability $P_{\rm LOS}(d_{aw})$. Next, the expected value of $P^*_{e,w}$ form Alice's perspective is characterized in the following theorem.
\begin{theorem}\label{col2}
The expected value of $P^*_{e,w}$ form Alice's perspective can be characterized as \eqref{Epew} shown at the top of the next page
where $P_{aw}({\rm L})=P_{\rm LOS}(d_{aw})$, $P_{aw}({\rm N})=1-P_{\rm LOS}(d_{aw})$, $\Gamma(\cdot)$ is the gamma function \cite[Eq. (8.310.1)]{gradshteyn2014table}, and $g_k^{(a,s)}$ and $b_k^{(a,s)}$ are defined above for $k\in\{1,2\}$. Moreover, the function $S(\nu_{\mathcal{B}},g_k^{(a,s)})$ is defined as
\begin{align}	\setcounter{equation}{11}
\!S(\nu_{\mathcal{B}},g_k^{(a,s)})\!\triangleq\!\sum_{l=1}^{\nu_{\mathcal{B}}}\!\binom{\nu_{\mathcal{B}}}{l}\!(-1)^l\left(\!1\!+\!l\frac{\eta_{\mathcal{B}}P_J^{\rm max}g_k^{(a,s)}}{P_am_{a,f}\nu_{\mathcal{B}}}\!\right)^{\!\!-\nu_{\mathcal{B}}}\!\!\!\!\!,
\end{align}
and
$I(\nu_{\mathcal{B}},l,g_k^{(a,s)})$ for $\nu_{\mathcal{B}}=1$ and $\nu_{\mathcal{B}}\geq2$ is defined as
\begin{align}
I(1,l,g_k^{(a,s)})&\triangleq\ln\left(1+l\frac{P_J^{\rm max}g_k^{(a,s)}}{P_am_{a,f}}\right),\label{I1}\\
I(\nu_{\mathcal{B}}\geq2,l,g_k^{(a,s)})&\triangleq\frac{(\nu_{\mathcal{B}}-2)!}{\nu_{\mathcal{B}}^{\nu_{\mathcal{B}}-1}}\bigg[1\nonumber\\
&\hspace{-0.5cm}-\left(1+l\frac{\eta_{\mathcal{B}}P_J^{\rm max}g_k^{(a,s)}}{P_am_{a,f}\nu_{\mathcal{B}}}\right)^{\!\!-\nu_{\mathcal{B}}+1}\bigg].
\end{align}
\end{theorem}
\begin{proof}
Let $P^{\rm C}_{e,w}$, $\lambda_1^{\rm C}$, and $\lambda_2^{\rm C}$ denote the values of $P^*_{e,w}$, $\lambda_1$, and $\lambda_2$, respectively, conditioned on the blockage instance $\mathcal{B}\in\{{\rm L,N}\}$ and the gain $g^{(a,s)}$ of Alice's second array to Willie. Then using \eqref{pe} we have
\begin{align}\label{Epew1}
\!\!\!\!\!\E[P^{\rm C}_{e,w}]\!\!&=\!\E_{\lambda_1^{\rm C}<\lambda_2^{\rm C}}[P^{\rm C}_{e,w}]\!\Pr(\lambda^{\rm C}_1\!\!<\!\!\lambda^{\rm C}_2)\!+\!\E_{\lambda^{\rm C}_1\geq\lambda^{\rm C}_2}[P^{\rm C}_{e,w}]\!\Pr(\lambda^{\rm C}_1\!\!\geq\!\!\lambda^{\rm C}_2)\!\nonumber\\
&\hspace{-0.7cm}=\Pr(\lambda^{\rm C}_1\!\geq\!\lambda^{\rm C}_2)\!\left(\!1\!-\!\frac{P_am_{a,f}}{P_J^{\rm max}g^{(a,s)}}\E_{\lambda^{\rm C}_1\geq\lambda^{\rm C}_2}\!\left[\!\frac{|\tilde{h}^{(\mathcal{B})}_{aw,f}|^2}{|\tilde{h}^{(\mathcal{B})}_{aw,s}|^2}\!\right]\!\right)\!.\!
\end{align}

The closed form of $\Pr(\lambda^{\rm C}_1\geq\lambda^{\rm C}_2)$ can be  derived as
\begin{align}\label{eqcolproof2}
\Pr(\lambda^{\rm C}_1\!\geq\!\lambda^{\rm C}_2)&=\Pr\left(|\tilde{h}^{(\mathcal{B})}_{aw,f}|^2\leq\frac{P_J^{\rm max}g^{(a,s)}}{P_am_{a,f}}|\tilde{h}^{(\mathcal{B})}_{aw,s}|^2\right)\nonumber\\
&\hspace{-1.7cm}\stackrel{(a)}{=}\sum_{l=0}^{\nu_{\mathcal{B}}}\!\binom{\nu_{\mathcal{B}}}{l}\!(-1)^l\E_{|\tilde{h}^{(\mathcal{B})}_{aw,s}|^2}\!\!\left[\exp\!\left(\!-\eta_{\mathcal{B}}l\frac{P_J^{\rm max}g^{(a,s)}}{P_am_{a,f}}|\tilde{h}^{(\mathcal{B})}_{aw,s}|^2\!\right)\!\right]
\nonumber\\
&\hspace{-1.7cm}\stackrel{(b)}{=}\sum_{l=0}^{\nu_{\mathcal{B}}}\!\binom{\nu_{\mathcal{B}}}{l}\!(-1)^l\left(1+l\frac{\eta_{\mathcal{B}}P_J^{\rm max}g^{(a,s)}}{P_am_{a,f}\nu_{\mathcal{B}}}\right)^{\!\!-\nu_{\mathcal{B}}},
\end{align}
where step $(a)$ follows from Alzer's lemma \cite{alzer1997some}, \cite[Lemma 6]{bai2015coverage} for a normalized gamma RV $X\sim{\rm Gamma}(\nu_{\mathcal{B}},1/\nu_{\mathcal{B}})$, which states that $\Pr\left(X<x\right)$ can tightly be approximated with $\left[1-\exp(-\eta_{\mathcal{B}} x)\right]^{\nu_{\mathcal{B}}}$ where $\eta_{{\mathcal{B}}}=\nu_{\mathcal{B}}(\nu_{\mathcal{B}}!)^{-1/\nu_{\mathcal{B}}}$, and then applying the binomial theorem assuming $\nu_{\mathcal{B}}$ is an integer \cite{bai2015coverage}. Moreover, step $(b)$ is derived using the moment generating function (MGF) of a normalized gamma RV $X$, i.e., $\E[{\rm e}^{tX}]=(1-t/\nu_{\mathcal{B}})^{-\nu_{\mathcal{B}}}$ for any $t<\nu_{\mathcal{B}}$.

Moreover, for the expectation term in \eqref{Epew1} we have
\begin{align}\label{eqcolproof3}
&\!\!\E_{\lambda^{\rm C}_1\geq\lambda^{\rm C}_2}\!\!\left[\!\frac{|\tilde{h}^{(\mathcal{B})}_{aw,f}|^2}{|\tilde{h}^{(\mathcal{B})}_{aw,s}|^2}\!\right]\!\!\!=\!\E\!\left[\!\frac{|\tilde{h}^{(\mathcal{B})}_{aw,f}|^2}{|\tilde{h}^{(\mathcal{B})}_{aw,s}|^2}{\Bigg|}|\tilde{h}^{(\mathcal{B})}_{aw,f}|^2\!\leq\!\!\frac{P_J^{\rm max}g^{(a,s)}}{P_am_{a,f}}|\tilde{h}^{(\mathcal{B})}_{aw,s}|^2\!\right]\nonumber\\
&{\hspace{0.5cm}}=\int_{0}^{\infty}\frac{f_{|\tilde{h}^{(\mathcal{B})}_{aw,s}|^2}(y)}{y}\Bigg[\underbrace{\int_{0}^{C_1y}xf_{|\tilde{h}^{(\mathcal{B})}_{aw,f}|^2}(x)dx}_{V_1}\Bigg]dy,
\end{align}
where $C_1\triangleq\frac{P_J^{\rm max}g^{(a,s)}}{P_am_{a,f}}$, and $f_{|\tilde{h}^{(\mathcal{B})}_{aw,f}|^2}(x)$ and $f_{|\tilde{h}^{(\mathcal{B})}_{aw,s}|^2}(y)$ are the probability density functions of the fading coefficients $|\tilde{h}^{(\mathcal{B})}_{aw,f}|^2$ and $|\tilde{h}^{(\mathcal{B})}_{aw,s}|^2$, respectively. Applying the part-by-part integration rule to $V_1$ and then using Alzer's lemma together with the binomial theorem yields
\begin{align}\label{eqcolproof4}
V_1=&C_1y\sum_{l_1=0}^{\nu_{\mathcal{B}}}\!\binom{\nu_{\mathcal{B}}}{l_1}\!(-1)^{l_1}{\rm e}^{-l_1\eta_{\mathcal{B}}C_1y}\nonumber\\
&-C_1y-\sum_{l_2=1}^{\nu_{\mathcal{B}}}\!\binom{\nu_{\mathcal{B}}}{l_2}\frac{(-1)^{l_2}}{\eta_{\mathcal{B}}l_2}\left[1-{\rm e}^{-l_2\eta_{\mathcal{B}}C_1y}\right].
\end{align}
By plugging \eqref{eqcolproof4} into \eqref{eqcolproof3}, using the MGF of the normalized gamma RV $|\tilde{h}^{(\mathcal{B})}_{aw,s}|^2$, and then noting that $f_{|\tilde{h}^{(\mathcal{B})}_{aw,s}|^2}(y)=\nu_{\mathcal{B}}^{\nu_{\mathcal{B}}}y^{\nu_{\mathcal{B}}-1}{\rm e}^{-\nu_{\mathcal{B}}y}/\Gamma(\nu_{\mathcal{B}})$ we have
\begin{align}\label{eqcolproof5}
&\E_{\lambda^{\rm C}_1\geq\lambda^{\rm C}_2}\!\left[\!\frac{|\tilde{h}^{(\mathcal{B})}_{aw,f}|^2}{|\tilde{h}^{(\mathcal{B})}_{aw,s}|^2}\!\right]=C_1\sum_{l_1=1}^{\nu_{\mathcal{B}}}\binom{\nu_{\mathcal{B}}}{l_1}\!(-1)^{l_1}\left(1+l_1\frac{\eta_{\mathcal{B}}C_1}{\nu_{\mathcal{B}}}\right)^{\!\!-\nu_{\mathcal{B}}}\nonumber\\
&-\sum_{l_2=1}^{\nu_{\mathcal{B}}}\!\binom{\nu_{\mathcal{B}}}{l_2}\frac{(-1)^{l_2}\nu_{\mathcal{B}}^{\nu_{\mathcal{B}}}}{\eta_{\mathcal{B}}l_2\Gamma(\nu_{\mathcal{B}})}\bigg[\int_{0}^{\infty}{y^{\nu_{\mathcal{B}}-2}}{\rm e}^{-\nu_{\mathcal{B}}y}dy\nonumber\\
&\hspace{3cm}-\int_{0}^{\infty}{y^{\nu_{\mathcal{B}}-2}}{\rm e}^{-(l_2\eta_{\mathcal{B}}C_1+\nu_{\mathcal{B}})y}dy\bigg].
\end{align}
Now given that the parameter $\nu_{\mathcal{B}}$ of Nakagami-$m$ fading is always larger than or equal  to $0.5$ and it is assumed an integer here, we have $\nu_{\mathcal{B}}\in\mathbb{N}$ where $\mathbb{N}$  stands for the set of natural numbers.
For $\nu_{\mathcal{B}}\geq2$ using \cite[Eq. (3.351.3)]{gradshteyn2014table} we have $\int_{0}^{\infty}y^{\nu_{\mathcal{B}}-2}{\rm e}^{-\alpha y}dy=(\nu_{\mathcal{B}}-2)!/\alpha^{\nu_{\mathcal{B}}-1}$ for any real $\alpha>0$. On the other hand, for $\nu_{\mathcal{B}}=1$ using \cite[Eq. (2.325.1)]{gradshteyn2014table} we have $\int_{0}^{\infty}y^{-1}{\rm e}^{-\alpha y}dy={\rm Ei}(-\alpha y)|_0^{\infty}$, where ${\rm Ei}(\cdot)$ is the exponential integral function defined as \cite[Eq. (8.211.1)]{gradshteyn2014table} for negative arguments. Therefore, following a similar approach to the proof of \cite[Corollary 2]{jamali2019uplink} we can calculate the difference of the two integrals in \eqref{eqcolproof5} as $\lim_{y\to0}[{\rm Ei}(-(l_2\eta_{\mathcal{B}}C_1+\nu_{\mathcal{B}})y)-{\rm Ei}(-\nu_{\mathcal{B}}y)]=\ln([l_2\eta_{\mathcal{B}}C_1+\nu_{\mathcal{B}}]/\nu_{\mathcal{B}})$ which is equal to $\ln(1+l_2C_1)$ for $\nu_{\mathcal{B}}=1$ (note that $\eta_{\mathcal{B}}=1$ for $\nu_{\mathcal{B}}=1$, and ${\rm Ei}(-\infty)=0$). This completes the proof of the theorem given  the definition of $I(\nu_{\mathcal{B}},l,g^{(a,s)})$ in \Tref{col2}. 
\end{proof}

\noindent\textbf{Remark 2.} In the derivation of \Tref{col2} it is assumed that Willie is not in the main lobe of Alice's first antenna array and hence receives the covert signal by a side lobe gain $m_{a,f}$. However, if Willie is within the main lobe of the first array, we should include another averaging over the gain $g^{(a,f)}$ of the first array given the beamsteering error, i.e., that gain is either $g_1^{(a,f)}\triangleq M_{a,f}$ with probability $b_1^{(a,f)}\triangleq F_{|\mathcal{E}_{a,f}|}\left({\theta_{a,f}}/{2}\right)$ or $g_2^{(a,f)}\triangleq m_{a,f}$ with probability $b_2^{(a,f)}\triangleq 1-F_{|\mathcal{E}_{a,f}|}\left({\theta_{a,f}}/{2}\right)$.
\vspace{-0.4cm}
\section{Performance of the  Alice-Bob Link}
\vspace{-0.05cm}
\subsection{Outage Probability}
\vspace{-0.05cm}
We assume that Alice targets a rate $R_b$ requiring the Alice-Bob link to meet a threshold signal-to-interference-plus-noise ratio (SINR) $\gamma_{\rm th}\triangleq2^{R_b}-1$. Then the outage probability $P_{\rm out}^{\rm AB}\triangleq\Pr({\gamma_{ab}}<\gamma_{\rm th})$ in achieving $R_b$ can be characterized as \Tref{thm3}, where the SINR $\gamma_{ab}$ of the  Alice-Bob link is given using \eqref{yb} as
\begin{align}\label{gab}
\gamma_{ab}=\frac{{P_aG_{ab,f}L_{ab}}|\tilde{h}_{ab,f}|^2}{{P_JG_{ab,s}L_{ab}}|\tilde{h}_{ab,s}|^2+\sigma^2_b}.
\end{align}
Note that in addition to  $|\tilde{h}_{ab,f}|^2$, $|\tilde{h}_{ab,s}|^2$ and $P_J$, the blockage instance $\mathcal{B}\in\{{\rm L},{\rm N}\}$ and the antenna gains can also randomly change from a transmission block to another block. In particular, while we assume that the jamming signal arrives with the deterministic side lobe gain $m_{a,s}$, there are still uncertainties in the gains of Alice's first array and Bob's receiver (they are pointing their main lobes) due to the beamsteering error. Therefore, the gain $g^{(a,f)}$ of the main lobe of Alice's first array pointed to Bob is either $g_1^{(a,f)}\triangleq M_{a,f}$ with probability $b_1^{(a,f)}\triangleq F_{|\mathcal{E}_{a,f}|}\left({\theta_{a,f}}/{2}\right)$ or $g_2^{(a,f)}\triangleq m_{a,f}$ with probability $b_2^{(a,f)}\triangleq 1-F_{|\mathcal{E}_{a,f}|}\left({\theta_{a,f}}/{2}\right)$. Similarly, the gain $g^{(b)}$ of Bob's receiver is either $g_1^{(b)}\triangleq M_{b}$ with probability $b_1^{(b)}\triangleq F_{|\mathcal{E}_{b}|}\left({\theta_{b}}/{2}\right)$ or $g_2^{(b)}\triangleq m_{b}$ with probability $b_2^{(b)}\triangleq 1-F_{|\mathcal{E}_{b}|}\left({\theta_{b}}/{2}\right)$. Furthermore, in the derivation of \Tref{thm3} we assume that Willie is not in the main lobe of Alice's first array. However, if Willie is in the Alice-Bob direction, we should include another averaging of the gain of the main lobe of Alice's second array carrying the jammer signal, i.e., instead of a deterministic $m_{a,s}$ we should consider two possibilities $g_k^{(a,s)}$ with probabilities $b_k^{(a,s)}$, $k\in\{1,2\}$, defined in Section III-B.

\begin{theorem}\label{thm3}
The outage probability of the Alice-Bob link in achieving the target rate $R_b\triangleq\log_2(1+\gamma_{\rm th})$ can be characterized as
\begin{align}\label{pout}
&P_{\rm out}^{\rm AB}\!=\!\!\!\!\sum_{\mathcal{B}\in\{{\rm L,N}\}}\!\!\!\!P_{ab}(\mathcal{B})\sum_{k_1=1}^{2}b_{k_1}^{(a,f)}\sum_{k_2=1}^{2}b_{k_2}^{(b)}\bigg[1\!+\!\sum_{l=1}^{\nu_{\mathcal{B}}}\!\binom{\nu_{\mathcal{B}}}{l}(-1)^{l}\nonumber\\
&\hspace{1cm}\times\exp\!\bigg(\!-\frac{l\eta_{\mathcal{B}}\gamma_{\rm th}\sigma^2_b}{P_ag_{k_1}^{(a,f)}g^{(b)}_{k_2}L_{ab}^{(\mathcal{B})}}\bigg)V(\nu_{\mathcal{B}},l,g^{(a,f)}_{k_1})\bigg],
\end{align}
where $P_{ab}({\rm L})=P_{\rm LOS}(d_{ab})$ and $P_{ab}({\rm N})=1-P_{\rm LOS}(d_{ab})$. Moreover,
$V(\nu_{\mathcal{B}},l,g^{(a,f)}_{k_1})$ for $\nu_{\mathcal{B}}=1$ and $\nu_{\mathcal{B}}\geq2$ is defined as
\begin{align}
&\hspace{-2.8cm}\!V(1,l,g^{(a,f)}_{k_1})\!\triangleq\!\frac{P_ag^{(a,f)}_{k_1}}{P_J^{\rm max}l\gamma_{\rm th}m_{a,s}}\ln\!\bigg(\!1\!+\!\frac{P_J^{\rm max}l\gamma_{\rm th}m_{a,s}}{P_ag^{(a,f)}_{k_1}}\!\bigg),\!\\
V(\nu_{\mathcal{B}}\geq2,l,g^{(a,f)}_{k_1})&\triangleq\frac{\nu_{\mathcal{B}}P_ag^{(a,f)}_{k_1}}{P_J^{\rm max}l\eta_{\mathcal{B}}\gamma_{\rm th}m_{a,s}(\nu_{\mathcal{B}}-1)}\nonumber\\
&\hspace{-0.7cm}\times\bigg[1-\bigg(1+\frac{P_J^{\rm max}l\eta_{\mathcal{B}}\gamma_{\rm th}m_{a,s}}{\nu_{\mathcal{B}}P_ag^{(a,f)}_{k_1}}\bigg)^{\!\!1-\nu_{\mathcal{B}}}\bigg].
\end{align}
\end{theorem}
\begin{proof}
Given the SINR of the Alice-Bob link in \eqref{gab}, we have for the outage probability conditioned on the blockage instance $\mathcal{B}$, and the antenna gains $g^{(a,f)}$ and $g^{(b)}$
\begin{align}\label{proofthm3-1}
P_{\rm out,C}^{\rm AB}\triangleq&\Pr(\gamma_{ab}<\gamma_{\rm th}|\mathcal{B},g^{(a,f)},g^{(b)})\nonumber\\
\stackrel{(a)}{=}&\Pr\left(|\tilde{h}_{ab,f}^{(\mathcal{B})}|^2<C_2P_J|\tilde{h}_{ab,s}^{(\mathcal{B})}|^2+C_3\right)\nonumber\\
&\hspace{-1.4cm}\stackrel{(b)}{=}\sum_{l=0}^{\nu_{\mathcal{B}}}\binom{\nu_{\mathcal{B}}}{l}(-1)^{l}{\rm e}^{-l\eta_{\mathcal{B}}C_3}\E_{P_J,|\tilde{h}^{(\mathcal{B})}_{ab,s}|^2}\!\!\left[{\rm e}^{-l\eta_{\mathcal{B}}C_2P_J|\tilde{h}^{(\mathcal{B})}_{ab,s}|^2}\!\right]\nonumber\\
&\hspace{-1.4cm}\stackrel{(c)}{=}\sum_{l=0}^{\nu_{\mathcal{B}}}\binom{\nu_{\mathcal{B}}}{l}(-1)^{l}{\rm e}^{-l\eta_{\mathcal{B}}C_3}\E_{P_J}\!\!\left[\!\left(1+\frac{l\eta_{\mathcal{B}}C_2P_J}{\nu_{\mathcal{B}}}\right)^{\!\!-\nu_{\mathcal{B}}}\right]
\nonumber\\
&\hspace{-1.4cm}\stackrel{(d)}{=}\!1\!+\!\sum_{l=1}^{\nu_{\mathcal{B}}}\!\binom{\nu_{\mathcal{B}}}{l}\!\frac{(-1)^{l}{\rm e}^{-l\eta_{\mathcal{B}}C_3}}{P_J^{\rm max}}\!\!\int_{0}^{P_J^{\rm max}}\!\!\!\!\!\!\!\Big(\!1\!+\!\frac{l\eta_{\mathcal{B}}C_2x}{\nu_{\mathcal{B}}}\Big)^{\!\!-\nu_{\mathcal{B}}}\!\!dx.\!\!
\end{align}
where in step $(a)$ we have defined $C_2\triangleq \gamma_{\rm th}m_{a,s}/(P_ag^{(a,f)})$ and $C_3\triangleq \gamma_{\rm th}\sigma^2_b/(P_ag^{(a,f)}g^{(b)}L_{ab}^{(\mathcal{B})})$. Moreover, step $(b)$ follows by Alzer's lemma together with the binomial theorem, and step $(c)$ is derived using the MGF of the normalized gamma RV $|\tilde{h}^{(\mathcal{B})}_{ab,s}|^2$. Finally,  taking the integral  in step $(d)$ and recalling the definition of  the function $V(\nu_{\mathcal{B}},l,g^{(a,f)}_{k_1})$ from the statement of the theorem complete the proof. 
\end{proof}
\subsection{Maximum Effective Covert Rate}
Given any target data rate $R_b$, Alice and Bob can have the effective communication rate $\overline{R}_{a,b}\triangleq R_b(1-P_{\rm out}^{\rm AB})$, where their outage probability $P_{\rm out}^{\rm AB}$ is obtained using \Tref{thm3}. The goal here is to determine the optimal value of $P_J^{\rm max}$ that maximizes $\overline{R}_{a,b}$ while also satisfying the covertness requirement, i.e., $\E[P^*_{e,w}]\geq 1-\epsilon$ for any $\epsilon>0$. Given that $\E[P^*_{e,w}]$ and $P_{\rm out}^{\rm AB}$  monotonically increase with $P_J^{\rm max}$ (hence $\overline{R}_{a,b}$ decreases with $P_J^{\rm max}$), we have the following proposition for the maximum covert communication rate that can be achieved in our setup. However, the optimal rate per \Pref{prop4} needs to be evaluated numerically. 

\begin{proposition}\label{prop4}
The maximum covert rate achievable in the considered setup, given fixed system and channel parameters and fixed covertness requirement $\epsilon$ and target data rate $R_b$, can be obtained as $\overline{R}^*_{a,b}\triangleq R_b(1-P_{\rm out}^{*\rm AB})$ where $P_{\rm out}^{*\rm AB}$ is defined as \eqref{Epew} while evaluated in $P_{J,\rm opt}^{\rm max}$ that is the solution of the equation $\E[P^*_{e,w}]=1-\epsilon$ for $P_J^{\rm max}$.
\end{proposition}

\noindent\textbf{Remark 3.} Since in the special case of $\nu_{\mathcal{B}}=1$ the normalized gamma distribution simplifies to the exponential distribution with mean one, we argue that the results derived in this paper particularize to the same system model under Rayleigh fading channels by substituting $\nu_{\mathcal{B}}=1$.

\section{Numerical Results}
In this section, we provide the numerical results for various performance metrics delineated in \Tref{col2}, \Tref{thm3}, and \Pref{prop4}. 
We consider the link lengths $d_{aw}=d_{ab}=25$ \si{m}, path loss exponents $\alpha_{\rm L}=2$, $\alpha_{\rm N}=4$,  path loss intercepts $C_{\rm L}=C_{\rm N}=10^{-7}$, Nakagami fading parameters $\nu_{\rm L}=3$, $\nu_{\rm N}=2$, main lobe gains $M_{a,f}=M_{a,s}=M_{b}=15$ \si{dB}, side lobe gains $m_{a,f}=m_{a,s}=m_{b}=-5$ \si{dB}, noise power $\sigma^2_w=\sigma^2_b=-74$ \si{dBm}, and array beamwidths $\theta_{a,f}=\theta_{a,s}=\theta_{b}=30^{\rm o}$, unless explicitly mentioned. 
We further assume that the beamsteering error follows a Gaussian distribution with mean zero and variance $\Delta^2$; hence, $F_{|\mathcal{E}|}(x)={\rm erf}(x/(\Delta\sqrt{2}))$ where ${\rm erf}({\cdot})$ denotes the error function \cite{di2015stochastic}. 
And we choose $\Delta=5^{\rm o}$ in our numerical analysis unless explicitly specified. 


\begin{figure}
	\centering
	\includegraphics[width=3.6in]{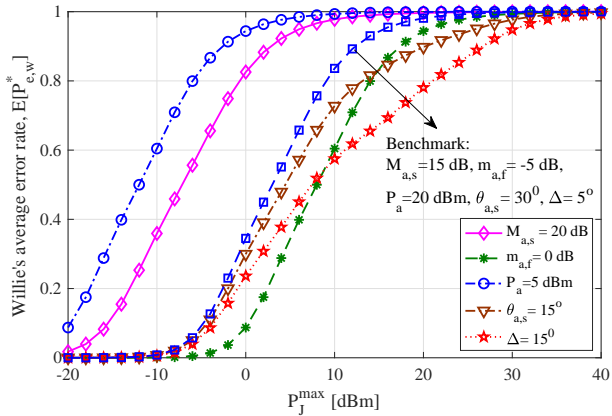}
	\caption{The expected value $\E[P^*_{e,w}]$ of Willie's detection error rate for a benchmark scenario with $M_{a,s}=15$ \si{dB}, $m_{a,f}=-5$ \si{dB}, $P_a=20$ \si{dBm}, $\theta_{a,s}=30^{\rm o}$, and  $\Delta=5^{\rm o}$. The effect of different parameters is explored by considering the values $M_{a,s}=20$ \si{dB}, $m_{a,f}=0$ \si{dB}, $P_a=5$ \si{dBm},  $\theta_{a,s}=15^{\rm o}$, and  $\Delta=15^{\rm o}$  while the rest of the parameters are exactly the same as the benchmark scenario.}
	\vspace{-0.15in}
\end{figure}
Figure 1 shows the expected value $\E[P^*_{e,w}]$ of Willie's detection error rate for various transceiver parameters. It is observed that $\E[P^*_{e,w}]$ monotonically increases with $P_J^{\rm max}$. Moreover, increasing $M_{a,s}$ and $\theta_{a,s}$ increases $\E[P^*_{e,w}]$ while increasing $m_{a,f}$, $P_a$, and $\Delta$ decreases $\E[P^*_{e,w}]$. 

\begin{figure}
	\centering
	\includegraphics[width=3.6in]{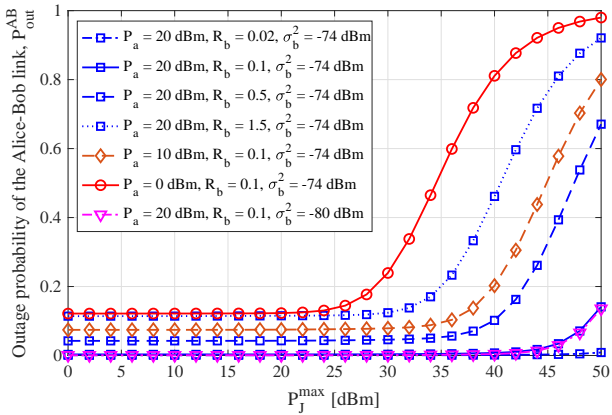}
	\caption{The outage probability of the Alice-Bob link for  various values of the transmit power $P_a$, threshold rate $R_b$, and noise variance $\sigma^2_b$.}
	\vspace{-0.15in}
\end{figure}
Figure 2 illustrates that the reliability of Alice-to-Bob transmission degrades by increasing the threshold rate $R_b$ and the noise variance $\sigma^2_b$ while increasing $P_a$ improves the performance. Also, $P_{\rm out}^{\rm AB}$ monotonically increases with $P_J^{\rm max}$.


Effective covert rates corresponding to the benchmark scenario in Figure 1 is summarized in Table I for $\epsilon=0.05$ and various threshold rates. It is observed that, for a given link,  the effective covert rate first increases and then decreases by increasing the threshold rate. Note that mmWave links benefit from much larger bandwidths compared to RF links; hence, the results in Table I imply much higher data rates, in bits per second.

\begin{table}[t]
	\centering
	\caption{Covert rates for $\epsilon=0.05$ and various threshold rates. }
	\label{T1}
	\begin{tabular}{M{0.7cm}||M{0.9cm}M{0.8cm}M{0.8cm}M{0.8cm}M{0.8cm}M{0.8cm}}  
		$R_b$ & $0.1$ & $0.5$& $1$& $2.5$ & $5$ & $10$  \\ \hline\hline 
		$P_{\rm out}^{*\rm AB}$ & $0.00313$ & $0.04253$ & $0.0935$ & $0.121$ & $0.1308$ & $0.9913$ \\ \hline
$\overline{R}^*_{a,b}$ & $0.0997$ &   $0.4787$ &   $0.9065$  &  $2.1975$  &  $4.3460$ &   $0.0870$
	\\ \hline 
	\end{tabular}
\vspace{-0.15in}
\end{table}

\section{Conclusion}
In this paper, we investigated covert communication over mmWave links. We employed a dual-beam transmitter to simultaneously transmit the desired signal to the destination and propagate a jamming signal to degrade the warden's  performance. We derived the closed-form expressions for the expected value of the warden's error rate and the outage probability of the Alice-Bob link. Our results demonstrated the superiority of mmWave links compared to RF links in terms of effective covert rates.





\end{document}